r
\documentclass[aps,pra,twocolumn,amsmath,showpacs,
amssymb,amsfonts,superscriptaddress]{revtex4}

\usepackage{graphics}
\usepackage{epsfig}
\usepackage{amsmath}

\newtheorem{theorem}{Theorem}

\newtheorem{propos}[theorem]{Proposition}
\newtheorem{lemma}[theorem]{Lemma}

\newenvironment{proof}[1][Proof]{\noindent\textbf{#1.} }{\ \rule{0.5em}{0.5em}}

\def\be{\begin{eqnarray}}
\def\bea{\begin{eqnarray}}
\def\bma{\begin{mathletters}}
\def\ee{\end{eqnarray}}
\def\eea{\end{eqnarray}}
\def\ema{\end{mathletters}}
\def\Tr{{\rm Tr}}

\begin{document}

\title{The decreasing property of relative entropy \\ and the strong superadditivity of quantum channels}

\author{Grigori G. Amosov}
\email{gramos@mail.ru}

\affiliation{Department of Higher Mathematics, Moscow Institute of Physics and Technology,
141700 Dolgoprudny,  Russia}

\author{Stefano Mancini}
\email{stefano.mancini@unicam.it}

\affiliation{Dipartimento di Fisica, Universit\`{a} di Camerino,
62032 Camerino, Italy\\
and INFN, Sezione di Perugia, 06123 Perugia, Italy}

\date{\today}

\begin{abstract}
We argue that a fundamental (conjectured) property of memoryless
quantum channels, namely the strong superadditivity, is
intimately related to the decreasing property of the quantum
relative entropy. Using the latter we first give, for a wide
class of input states, an estimation of the output entropy for
phase damping channels and some Weyl quantum channels. Then we
prove, without any input restriction,  the strong superadditivity
for several quantum channels, including depolarizing quantum
channels, quantum-classical channels and quantum erasure channels.
\end{abstract}


\maketitle


\section{Introduction}

The apparently simple concept of distinguishability is at the root of information
processing, even at the quantum ground.
For instance, it is rather intuitive that the amount of classical information (symbols encoded
into quantum states) that  can be reliably transmitted through a quantum channel will ultimately
depend upon the ability of the receiver to distinguish different quantum states.
Unlike with classical states, two different quantum states are not necessarily fully
distinguishable. In \cite{Ved02} it was argued that the quantum
relative entropy is the most appropriate quantity to measure
distinguishability between different quantum states.
Hence it could be a powerful tool for investigating quantum channels' properties.
The quantum relative entropy does not
increase under physical processes (completely and trace preserving maps) \cite{Lin75}.
Thus two states can only become less distinguishable
as they undergo any kind of physical transformation.
This result will be central to this paper.

There is a single quantity that completely characterize a quantum channel for transmitting classical
information: its classical capacity \cite{Hol72}. It represents the maximum rate
at which classical symbols can be transmitted through the channel in a reliably way.
It should thus come from the average over a large number (actually infinity) of channel uses.
However, it was conjectured that memoryless channels posses the nice
\emph{additivity property}, that is the classical capacity adds up with the number of channel uses \cite{Hol98, Sch97}.
Hence, it can be simply evaluated by considering one use (one shot)
of the channel, likewise in the classical case due to the Shannon coding theorem.
This has the profound implication that entangled inputs do not matter for the capacity of memoryless
quantum channels.
 The additive property has been proved for a class of quantum channels  \cite{Amo00, Kin02, Amo06a, Amo07b}
 and it was suspected that $l_p$-norms play a crucial role for the global proof.
Unfortunately, recently it has been shown that this is not the case \cite{Hay07}. Thus, the need to devise
 alternative methods.

In reality, the additivity property as discussed above, can be
traced back to the additivity of the minimal output entropy of
two channels. In contrast, when we consider the minimum of the
average output entropies, we are led to the \emph{superadditivity
property}. That is, the minimum of the average output entropies
for the tensor product of two quantum channels is greater than or
equal to the sum of the minima corresponding to the  single
channels. This property was conjectured in \cite{Hol04} and it
turns out to be stronger than the simple additive property. In
fact, if the strong superadditivity property holds, then the
additivity property follows \cite{Hol04}.

Thus, it is of uppermost importance to prove the strong superadditivity
for memoryless quantum channels. Actually, it has only been proved for
entanglement-breaking channels and noiseless channels \cite{Hol04}
and for the quantum depolarizing channel \cite{Amo07a} using different methods.

In the present paper we argue that the strong superadditivity is related to the decreasing
property of the relative entropy. Hence we shall provide a proof of the strong
superadditivity based on the decreasing property of the relative entropy for a class of quantum channels.
This class includes the above channels (thus giving an alternative proof) as well as others ones (thus representing an extension over the already know results).

The layout of the paper is the following. In Section II we recall
some basic notions about quantum relative entropy and classical
capacity of quantum channels. Section III is devoted to formalize
the additivity and the strong superadditivity properties. We give
some estimates of the output entropy for the phase damping
channels and for a subclass of Weyl channels in Section IV and
Section V respectively . Finally, in Section VII we prove the
strong superadditivity for a class of quantum channels without
any restriction on the input states. Section VII is for
conclusions.


\section{Basic Notions}

We start by recalling the definition of the von Neumann entropy of a quantum
system described by a density matrix ${\rho}$ belonging to the set of
states $\mathfrak{S}({\mathcal H})$ (positive unit trace operators)
of the Hilbert space ${\mathcal H}$ of dimension $d<+\infty$,
\begin{eqnarray}
S(\rho) := -\Tr( \rho \log \rho ),
\nonumber
\end{eqnarray}
which can be considered as the proper quantum
analogue of the Shannon entropy \cite{Ohy93}.

Moving on from Shannon relative entropy we can consider the
von Neumann relative entropy as well.
The von Neumann relative entropy
 between  the two states $\sigma$, $\rho$ $\in\mathfrak{S}({\mathcal H})$
 is defined as
\begin{eqnarray}
S(\sigma ||\rho) := \Tr\left[\sigma (\log \sigma - \log \rho)\right].
\nonumber
\end{eqnarray}
Actually, this
quantity was first
considered by Umegaki \cite{Ume62} and it is often referred to it as the Umegaki entropy.
This measure has the same statistical
interpretation as its classical analogue: it tells us how
difficult it is to distinguish the state $\sigma$ from the state
$\rho$ \cite{Hia91}.

Moreover, it has  three simple  properties:
\begin{itemize}

\item[i)] Unitary operations $U$ leave $S(\sigma||\rho)$
invariant, i.e.
$S(\sigma||\rho)=S(U\sigma U^{*}||U \rho U^{*})$. Unitary
transformations represent a change of basis
and the distance between two states should not  change under
this.

\item[ii)] $S(\Tr_p \sigma ||\Tr_p \rho) \le S({\sigma}||\rho)$,
where $\Tr_p$ is a partial trace. Tracing over a part of the system
leads to a loss of information. Hence, the less information we have about two states,
the harder they are to distinguish.

\item[iii)] The relative entropy is additive $S(\sigma_1 \otimes \sigma_2||\rho_1 \otimes
\rho_2)=S(\sigma_1||\rho_1)+S(\sigma_2||\rho_2)$. This inequality is a
consequence of additivity of entropy itself.

\end{itemize}
These properties have profound implication for the quantum states' transformation
(or quantum systems' evolution). In fact the following theorem holds  \cite{Lin75}:

\begin{theorem}[Decreasing property of relative entropy]
For any completely positive, trace preserving map
$\Phi:\mathfrak{S}({\mathcal H})\to\mathfrak{S}({\mathcal H})$
given by $\Phi(\sigma) = \sum_i A_i\sigma A^{*}_i$ such that $\sum
A^{*}_i A_i = 1$, we have
$$S(\Phi(\sigma) ||\Phi(\rho)) \le
S({\sigma}||\rho),$$
with $\sigma$, $\rho$ $\in\mathfrak{S}({\mathcal H})$.
\label{decr}
\end{theorem}

We simply present a physical argument as to why we should expect
this theorem to hold.
A completely positive map (CP-map) can be represented
as a unitary transformation on an extended Hilbert space. According to
i), unitary transformations do not change the relative entropy between
two states. However, after this, we have to perform a partial trace
to go back to the original Hilbert space which, according to ii),
decreases the relative entropy as some information is invariably lost
during this operation. Hence the relative entropy decreases under
any CP-map.

A simple consequence of the fact that
the quantum relative entropy itself does not increase under
CP-maps quantum distinguishability never increases.
Another consequence is that correlations (as measured by the quantum
mutual information) also cannot increase, but now
under {\em local} CP-maps.

In classical information theory the capacity for
communication is given by the mutual information between
sent message and received message \cite{Cov91}. This is intuitively clear, since mutual information
quantifies correlations between sent and received messages and it
thus tells us how faithful the transmission is. If we use quantum
states to encode symbols, then the capacity is not given by the
quantum mutual information, but is given by the so called HSW bound \cite{Hol98, Sch97}.

The linear map $\Phi : \mathfrak{S} ({\mathcal H})\to \mathfrak{S} ({\mathcal H})$ is said to be a
quantum channel if it is completely positive \cite{Hol72}.
Moreover, the quantum channel $\Phi $ is called bistochastic (or unital) if
$\Phi (\frac {1}{d}I_{{\mathcal H}})=\frac {1}{d}I_{{\mathcal H}}$, where $I_{{\mathcal H}}$ is
the identity operator in ${\mathcal H}$.

The HSW bound $C_{1}(\Phi)$ of a
quantum channel $\Phi $ is defined by the formula
\begin{equation}
C_1(\Phi ):=\sup \left[ S\left(\sum \limits
_{j=1}^r\pi_j\Phi (x_j)\right)-\sum \limits _{j=1}^r\pi_jS\left(\Phi (x_j)\right)\right],
\nonumber
\end{equation}
where  the supremum is taken over all probability distributions $\{\pi_j\}_{j=1}^r$
and states ${x_j\in \mathfrak{S} ({\mathcal H})}$.

Notice that
\begin{eqnarray}
&&S\left(\sum \limits_{j=1}^r\pi_j\Phi (x_j)\right)-\sum \limits _{j=1}^r\pi_jS\left(\Phi (x_j)\right)\nonumber\\
&=&\sum \limits_{j=1}^r\pi_jS\left(\Phi (x_j) \,\Big\|\, \sum \limits _{l=1}^r\pi_l\Phi (x_j)\right),\nonumber
\end{eqnarray}
so that we have a direct link to the relative entropy.

The additivity
conjecture states that for any two channels
$\Phi$ and $\Omega $
\begin{equation}
C_1(\Phi \otimes \Omega )=C_1(\Phi )+ C_1(\Omega ).
\nonumber
\end{equation}
If the additivity conjecture holds, one can easily find the
capacity $C(\Phi )$ of the channel $\Phi $ by the formula (see \cite {Hol98})
\begin{equation}
C(\Phi)=\lim \limits _{n\to +\infty } \frac {C_1(\Phi
^{\otimes n})}{n}=C_1(\Phi ).
\nonumber
\end{equation}


\section{The strong superadditivity}

Given a quantum channel $\Phi $ in a Hilbert space ${\mathcal H}$ let us put \cite{Hol04}
\begin {equation}\label {quant}
 H_{\Phi}(\rho):=\min \sum \limits
_{j=1}^{k}\pi _{j}S(\Phi (\rho _{j})),
\end {equation}
where $\rho=\sum \limits _{j=1}^{k}\pi _{j}\rho _{j}$ and
the minimum is taken over all probability distributions $\{\pi _{j}\}_{j=1}^k$ and states $\rho _{j}\in \mathfrak{S}({\mathcal H})$.

The strong superadditivity conjecture for the channel $\Phi$ states that
\begin {equation}\label {strong}
 H_{\Phi \otimes \Omega}(\rho)\ge
H_{\Phi}(\Tr_{{\mathcal K}}(\rho))+ H_{\Omega}(\Tr_{{\mathcal H}}(\rho)),
\end {equation}
with $\rho \in \mathfrak{S}({\mathcal H}\otimes {\mathcal K})$, for an arbitrary quantum channel $\Omega $ in the Hilbert space ${\mathcal K}$.

The infimum of the output entropy of a quantum channel $\Phi $ is
defined by
\begin {equation}\label {addit}
S_{min}(\Phi):=\inf \limits _{\rho\in \mathfrak{S}({\mathcal H})}S(\Phi
(\rho)).\nonumber
\end {equation}
The additivity conjecture for the quantity $S_{min} (\Phi)$ states that \cite{Hol98}
\begin {equation}\label {conj}
S_{min} (\Phi\otimes \Omega)=S_{min} (\Phi)+S_{min} (\Omega)
\end {equation}
for an arbitrary quantum channel $\Omega $. It was shown in \cite{Hol04} that if the strong superadditivity holds, then the additivity follows.
Hence, the conjecture (\ref {strong}) is stronger than
(\ref {conj}).

At first time the additivity property (\ref{conj}) was proved for
quantum depolarizing channel \cite{Kin02}. The method
was based upon the estimation of $l_{p}$-norms of the channel.
Since then, it was suspected that $l_p$-norms play a crucial role for the global proof.
Unfortunately, recently it has been shown that this is not the case \cite{Hay07}. Thus, the need to devise
 alternative methods.


\section{Estimation of the output entropy for the phase damping channel}

Let $\{|e_{s}\rangle\}_{s=0}^{d-1}$ and $\{\lambda _{s}\}_{s=0}^{d-1}$ be an
orthonormal basis in the Hilbert space ${\mathcal H}$ of dimension $d$ and a
probability distribution, respectively. Then, one can introduce the
unitary operator
$$
V:=\sum \limits _{s=0}^{d-1}\exp\left(i\frac {2\pi s}{d}\right)|e_{s}\rangle\langle e_{s}|,
$$
so to define the phase damping channel as
\begin {equation}\label {phasdam}
\Phi (\rho):=\sum \limits _{j=0}^{d-1}\lambda _{j}V^{j}\rho V^{*j},
\end {equation}
where $\rho \in {\mathfrak S}({\mathcal H})$. The numbers $\{\lambda _{s}\}$ give
the {\it spectrum} of the phase damping channel $\Phi$.
Furthermore, the completely positive map defined as
$$
E(\rho):=\frac {1}{d}\sum \limits _{j=0}^{d-1}V^{j}\rho
V^{*j}=\sum \limits _{s=0}^{d-1}|e_{s}\rangle\langle e_{s}|\rho |e_{s}\rangle\langle e_{s}|,
$$
represents the conditional expectation on the
algebra of fixed elements of $\Phi$.

We shall call a pure state $\rho =|f\rangle \langle f| \rangle\in
{\mathfrak S}({\mathcal H})$ \emph{unbiased} with respect to the
basis $\{|e_{s}\rangle\}$ if
\begin {equation}\label {unb}
\Tr(\rho |e_{s}\rangle\langle e_{s}|)=\frac {1}{d},\quad 0\le s\le
d-1. \nonumber
\end {equation}
The above condition is equivalent to the property
\begin {equation}\label {unb2}
|\langle \psi |e_{s}\rangle |=\frac {1}{\sqrt d},\quad 0\le s\le
d-1.
\end {equation}
Notice that if (\ref {unb2}) is satisfied for vectors
$|f\rangle=|f_{j}\rangle$, $0\le j\le d-1$ forming an orthonormal basis in ${\mathcal H}$,
then the bases $\{|f_{j}\rangle\}$ and $\{|e_{s}\rangle\}$ are said to be \emph{mutually
unbiased} \cite{Iva81}.

Let us denote by $\mathcal A$ a convex set of states which can be
represented as a convex linear combination of pure states $\rho
=|f \rangle \langle f |$ being unbiased with respect to the basis
$\{|e_s\rangle\}$ (eigenvectors of the unitary operators
introduced in the definition of the phase damping channel  (\ref
{phasdam})). As a consequence ${\mathcal A}$ is a convex set.
Moreover the following proposition holds.

\begin{propos}\label {fff1}
Suppose that $\rho \in {\mathcal A}$, then for the phase damping
(\ref {phasdam}) we get
$$
H_{\Phi}(\rho )\le -\sum \limits _{j=0}^{d-1}\lambda _{j}\log
\lambda _{j}.
$$
\end{propos}

\begin{proof}[Proof Proposition \ref{fff1} ]
Given $\rho \in {\mathcal A}$ we can write it as the convex
linear combination $\rho =\sum \limits _{k}\pi_{k}\rho _{k},\ \rho
_{k}=|f_{k}\rangle \langle f_{k}|\in {\mathcal A}$ such that
$$
S(\rho _{k})=-\sum \limits _{j=0}^{d-1}\lambda _{j}\log \lambda
_{j}.
$$
Thus, the result follows from the definition of $H_{\Phi}(\rho)$.
\hfill\end{proof}

\begin{propos} Suppose that for $\rho \in \mathfrak{S}({\mathcal H}\otimes {\mathcal K})$ the
following inclusion holds,
$$
\Tr_{{\mathcal K}}(\rho )\in {\mathcal A}.
$$
Then,
\begin {equation}\label {Est}
S((\Phi \otimes Id)(\rho ))\ge -\sum \limits _{j=0}^{d-1}\lambda
_{j}\log \lambda _{j}+\frac {1}{d}\sum \limits
_{j=0}^{d-1}S(\rho_{j}),
\nonumber
\end {equation}
where $\rho_{j}=d\,\Tr_{{\mathcal H}}((|e_{j} \rangle\langle e_{j}|\otimes
I_{{\mathcal K}})\rho )\in \mathfrak{S} ({\mathcal K})$. \label{estimdep}
\end{propos}

\begin{proof}[Proof Proposition \ref{estimdep} ]
The proof treads \cite {Amo07b} steps.
Let us take $\rho \in \mathfrak{S}({\mathcal H}\otimes {\mathcal K})$ such that
$\Tr_{{\mathcal K}}(\rho )\in {\mathcal A}$ and define
a quantum channel $\Xi _{\rho}:\mathfrak {S}({\mathcal H}\otimes {\mathcal K})\to
\mathfrak {S}({\mathcal H}\otimes {\mathcal K})$ by the formula
$$
\Xi _{\rho}(\sigma):=\sum \limits_{j=0}^{d-1}
\Tr((|e_{j}\rangle\langle e_{j}|\otimes I_{{\mathcal K}})\sigma )(V^{j}\otimes
I_{{\mathcal K}})\rho (V^{*j}\otimes I_{{\mathcal K}}),
$$
with $\sigma \in \mathfrak {S}({\mathcal H}\otimes {\mathcal K}). $ Then, let be
\begin{eqnarray}
\sigma &=& \sum \limits _{j=0}^{d-1}\lambda _{j}|e_{j}\rangle\langle e_{j}|\otimes y,\nonumber\\
\overline \sigma &=& \sum \limits _{j=0}^{d-1}\frac
{1}{d}|e_{j}\rangle\langle e_{j}|\otimes y\equiv\frac {1}{d}I_{{\mathcal H}}\otimes y,\nonumber
\end{eqnarray}
with $y\in \mathfrak {S}({\mathcal K})$ an arbitrary fixed state. It follows
$$
\Xi _{\rho}(\sigma)=(\Phi \otimes Id)(\rho),
$$
$$
\Xi _{\rho}(\overline \sigma)=\frac {1}{d}\sum \limits
_{j=0}^{d-1}(V^{j}\otimes I_{{\mathcal K}})\rho (V^{*j}\otimes I_{{\mathcal K}})\equiv
\tilde E(\overline\sigma).
$$
Here and throughout the paper $Id$ denotes the identity map.
Also notice that $\tilde E=(E\otimes Id)$ is the conditional
expectation to algebra of the elements being fixed with respect to
the action of the cyclic group $\{V^{j}\otimes I_{{\mathcal K}} ,\ 0\le j\le
d-1 \}$.

Now, on the one hand, Theorem \ref{decr} gives us
\begin {equation}\label {EE2}
S\left(\Xi_{\rho}(\sigma)\big\| \Xi _{\rho}(\overline \sigma)\right)\le S(\sigma
\|\overline \sigma)=\sum \limits _{j=0}^{d-1}\lambda
_{j}\log\lambda _{j}+\log d.
\end {equation}

On the other hand, it is
\begin{eqnarray}
S\left(\Xi _{\rho}(\sigma)\big\|\Xi _{\rho}(\overline \sigma)\right)&=&
\Tr((\Phi\otimes Id)(\rho )\log(\Phi \otimes Id)(\rho))\nonumber\\
&&-\Tr((\Phi \otimes Id)(\rho)\log \tilde E(\rho))\nonumber\\
&=&-S((\Phi \otimes Id)(\rho))\nonumber\\
&&-\Tr(\tilde E\circ (\Phi \otimes Id)(\rho)\log \tilde E(\rho))\nonumber\\
&=&-S((\Phi \otimes Id)(\rho))+S(\tilde E(\rho)).
\label {EE1}
\end {eqnarray}
In the above equations, we have used the equality $\tilde E\circ (\Phi \otimes Id)=\tilde E$
which holds because $\tilde E$ is the conditional expectation
to the algebra of elements being fixed with respect to the action
of $\Phi \otimes Id$.

Since
$$
\tilde E(\rho)=\frac {1}{d}\sum \limits
_{j=0}^{d-1}|e_{j}\rangle\langle e_{j}|\otimes \rho_{j},\quad \rho_{j}\in \sigma
({\mathcal K}),
$$
it follows
\begin {equation}\label {EE3}
S(\tilde E(\rho))=\log d+\frac {1}{d}\sum \limits
_{j=0}^{d-1}S(\rho_{j}),
\end {equation}
with $\rho_{j}=d\, \Tr_{{\mathcal H}}((|e_{j}\rangle\langle e_{j}|\otimes I_{{\mathcal K}})\rho),\ 0\le j\le
d-1$. Then, combining (\ref {EE2}), (\ref {EE1}) and (\ref {EE3}) we
get the result of the proposition \ref{estimdep}.
\hfill\end{proof}

\bigskip

We can now single out a wide class (over the totality) of input
states for which the phase damping channels respect
a kind of superadditivity property.

\begin{theorem}
Suppose that $\rho \in \mathfrak{S}({\mathcal H}\otimes {\mathcal K})$ is such that
$$
\Tr_{{\mathcal K}}(\rho )\in {\mathcal A}.
$$
Let  $\Phi$ be the phase damping channel \eqref{phasdam}, then the inequality
\begin{eqnarray}
S((\Phi \otimes \Omega)(\rho))&\ge& -\sum \limits
_{j=0}^{d-1}\lambda _{j}\log\lambda _{j}+ H_{\Omega
}(\Tr_{{\mathcal H}}(\rho))\nonumber\\
&\ge& H_{\Phi}(\Tr_{{\mathcal K}}(\rho ))+H_{\Omega }(\Tr_{{\mathcal
H}}(\rho)),\nonumber
\end{eqnarray}
holds for an arbitrary quantum channel $\Omega :{\mathfrak
S}({\mathcal K})\to {\mathfrak S}({\mathcal K})$.
\label {fff2}
\end{theorem}


\begin{proof}[Proof Theorem \ref{fff2}]
Defining $\tilde \rho :=(Id\otimes \Omega )(\rho )$, we notice
that  $\Tr_{{\mathcal K}}(\tilde \rho )\in {\mathcal A}$ and
\begin {equation}\label {xxx1}
S((\Phi \otimes \Omega )(\rho))=S((\Phi \otimes Id)(\tilde \rho
)). \nonumber
\end {equation}
Applying the Proposition \ref{estimdep} we obtain
\begin{equation}\label {xxx2}
S((\Phi \otimes \Omega )(\rho))\ge -\sum \limits
_{j=0}^{d-1}\lambda _{j}\log \lambda _{j}+\frac {1}{d}\sum
\limits _{j=0}^{d-1}S(\rho_{j}),
\end {equation}
where $\rho_{j}=d\,\Tr_{{\mathcal H}}((|e_{j} \rangle\langle
e_{j}|\otimes I_{{\mathcal K}})(Id \otimes \Omega )(\rho ))\in
\mathfrak{S}({\mathcal K})$. Using Proposition \ref {fff1} we can
rewrite (\ref {xxx2}) as
$$
S((\Phi \otimes \Omega )(\rho))\ge H_{\Phi}(\Tr_{{\mathcal
K}}(\rho ))+\frac {1}{d}\sum \limits _{j=0}^{d-1}S(\Omega (\rho
_{j})).
$$
Finally, taking into account that $\frac {1}{d}\sum \limits
_{j=0}^{d-1}\rho _{j}=\Omega (\Tr_{\mathcal H}(\rho))$, we obtain
$$
\sum \limits _{j=0}^{d-1}S(\Tr_{{\mathcal H}}(\rho _{j}))\ge
H_{\Omega}(\Tr_{{\mathcal H}}(\rho )).
$$
The result of the theorem \ref{fff2} then follows.
\hfill\end{proof}


\section{Estimation of the Output Entropy for the Weyl Channels}

Let us consider an orthonormal basis $|k\rangle,\ k=0,1,\dots ,d-1$ of the Hilbert
space ${\mathcal H}$ of dimension $d$ and define the unitary operators
\begin{equation}
U_{m,n}:=\sum \limits _{k=0}^{d-1}e^{\frac {2\pi i}{d}kn}|k\oplus m\rangle\langle k|,
\label{Wops}
\end{equation}
where $0\le m,n\le d-1$ and $\oplus$ denotes the sum modulus $d$. The operators \eqref{Wops} satisfy the Weyl
commutation relations
\begin {equation}\label {WO}
U_{m,n}U_{m',n'}=e^{2\pi i(m'n-mn')/d}U_{m',n'}U_{m,n},
\nonumber
\end {equation}
hence, we shall call them Weyl
operators. Notice that
\begin {equation}\label {WOA}
U_{m,0}|k\rangle=|k\oplus m \rangle,\quad U_{0,n}|k\rangle=e^{\frac {2\pi i}{d}kn}|k\rangle.
\nonumber
\end {equation}
We shall consider bistochastic quantum channels of
the following form
\begin {equation}\label {Weyl}
\Phi (\rho):=\sum \limits _{m,n=0}^{d-1}\pi _{m,n}U_{m,n}\,\rho \,U_{m,n}^{*},
\end {equation}
where $\{\pi_{m,n}\}_{m,n=0}^{d-1}$ are probability distributions and
$\rho\in \mathfrak{S}({\mathcal H})$ states. The channels (\ref {Weyl}) are called Weyl channels.

\bigskip

Now, let us fix positive numbers $0\le p_{n}, r_{m}  \le 1$,
$1\le n \le d-1$, $0\le m \le d-1$ such that $d\sum \limits
_{n=1}^{d-1}p_{n}+\sum \limits _{m=0}^{d-1}r_{m}=1$ and let us consider
the Weyl channel
\begin {equation}\label {Chan}
\Phi (\rho)=\sum \limits _{m=0}^{d-1}r_{m}U_{m,0}\,\rho \,U_{m,0}^{*}+\sum
\limits _{m=0}^{d-1}\sum \limits _{n=1}^{d-1}p_{n}
U_{m,n}\,\rho \,U_{m,n}^{*},
\end {equation}
$\rho\in \mathfrak{S}({\mathcal H})$.

It is shown in \cite {Amo07b} that the channels (\ref {Chan}) is
covariant with respect to the maximum commutative group of
unitary operators. Moreover, if the dimension of the space $d$ is
a prime number, the following decomposition holds
\begin {equation}\label {Equat}
\Phi (\rho)=\sum \limits _{k=0}^{d-1}\sum \limits
_{m=0}^{d-1}c_{m}U_{m,0}\Psi _{k}(\rho)U_{m,0}^{*},
\end {equation}
where $\rho\in \mathfrak{S}({\mathcal H})$ and
$$
\Psi _{k}(\rho)=\sum \limits _{n=0}^{d-1}\lambda _{n}U_{nk\ mod\
d,n}\,\rho \,U_{nk\ mod\ d,n}^{*},
$$
are phase damping channels. Furthermore, it is
\begin{eqnarray}
\lambda _{0}&=&1-d\sum \limits _{n=1}^{d-1}p_{n},\nonumber\\
\lambda_{n}&=&dp_{n},\quad 1\le n\le d-1,\nonumber\\
c_{m}&=&\frac {r_{m}}{d\left(1-d\sum \limits _{n=1}^{d-1}p_{n}\right)},\quad
0\le m\le d-1.\nonumber
\end{eqnarray}

\bigskip
We can now single out a wide class (over the totality) of  input states for which
the Weyl channels \eqref{Chan} respect a kind of superadditivity property.

Let us denote by $\mathcal A$ the maximum commutative algebra generated by
the projectors $|k\rangle\langle k|,\ 0\le k\le d-1$. Notice that the states $\rho \in {\mathcal A}$
are mutually unbiased with respect to the eigenvectors of the
unitary operators $U_{nk,n},\ 0\le k,n\le d-1$ \cite{Amo07b}.
Then, the following theorem holds.

\begin {theorem}\label {weyl}
Let the dimension $d$ of the space ${\mathcal H}$ be a prime number. Suppose
that $\rho \in {\mathfrak S}({\mathcal H}\otimes {\mathcal K})$ is such that
$$
\Tr_{{\mathcal K}}(\rho)\in {\mathcal A}.
$$
Let  $\Phi$ be the Weyl channel \eqref{Chan}, then the inequality
$$
S((\Phi \otimes \Omega)(\rho))\ge H_{\Phi}(\Tr_{{\mathcal K}}(\rho
))+H_{\Omega }(\Tr_{{\mathcal H}}(\rho)),
$$
holds for an arbitrary quantum channel $\Omega :{\mathfrak
S}({\mathcal K})\to {\mathfrak S}({\mathcal K})$.
\end {theorem}

\begin{proof}[Proof Theorem \ref{weyl}]
Using the decomposition (\ref {Equat}) we easily arrive at
\begin {equation}\label {lll}
S((\Phi \otimes \Omega )(\rho))\ge \frac {1}{d}\sum \limits
_{k=0}^{d-1}S((\Psi _{k}\otimes \Omega )(\rho )). \nonumber
\end {equation}
Then, by applying Theorem \ref {fff2} to each term of the right hand side
of (\ref {lll}) we obtain the result of Theorem \ref{weyl}.
\hfill\end{proof}


\section{Quantum Channels Respecting the Strong Superadditivity}

We shall provide hereafter a class of quantum channels that fully respect the strong superadditivity, i.e. without any restriction on the input states.

\subsection {The quantum noiseless channel}

The quantum noiseless channel in the
Hilbert space ${\mathcal H}$ of the dimension $d$
is simply defined as the identity operation
\begin{equation}
\Phi(\rho):=Id(\rho)=\rho,
\label{nc}
\end{equation}
with $\ \rho\in \mathfrak{S} ({\mathcal H})$.

\begin{theorem}
Let $\Phi$ be the quantum noiseless channel of Eq.(\ref{nc}), then the inequality
\begin {equation}
 H_{\Phi\otimes \Omega}(\rho)\ge  H_{\Omega}(\Tr_{{\mathcal H}}(\rho)),
 \nonumber
\end {equation}
 holds for an
arbitrary quantum channel
$\Omega :{\mathfrak
S}({\mathcal K})\to {\mathfrak S}({\mathcal K})$.
\label{Identity}
\end{theorem}

\begin{proof}[Proof Theorem \ref{Identity}]
Actually this theorem was proved in \cite {Hol04}. Our prove is
alternative and based upon the decreasing property of the
relative entropy. Let us take the optimal ensemble $\{\rho
_{k}\}$ such that
$$
H_{\Phi\otimes \Omega}(\rho)=\sum \limits _{k}\pi _{k}S((\Phi
\otimes \Omega)(\rho _{k})).
$$
Given a state $\rho _{k}\in {\mathfrak S}({\mathcal H}\otimes
{\mathcal K})$, the identity channel can be considered as the
phase damping channel $\Psi $ with the spectrum $\lambda _{0}=1,\
\lambda _{j}=0,\ 1\le j\le d-1$, for which the state
$\Tr_{{\mathcal K}}(\rho)\in {\mathcal A}$, where $\mathcal A$ is
the convex set generated by pure states unbiased with respect to
the basis of eigenvectors of the unitary operator determining
$\Psi $. Hence, the result follows from Theorem \ref {fff2}.
\hfill\end{proof}


\subsection {The quantum-classical channel}

Let $\{ M_{j},\ 1\le j\le d\}$ be a resolution of the identity
in ${\mathcal H}$ consisting of positive operators $ M_{j}>0,\ \sum
\limits _{j=1}^{d} M_{j}=I_{{\mathcal H}}$. The quantum channel $\Phi $
is said to be a quantum-classical channel (shortly q-c channel) if there exists an orthogonal
basis $\{|e_{j}\rangle\}$ in ${\mathcal H}$ such that \cite{Hol04}:
\begin {equation}\label {qc}
\Phi ( \rho)=\sum \limits _{j=1}^{d}\Tr( M_{j}
\rho)|e_{j}\rangle\langle e_{j}|.
\end {equation}

\begin{theorem}
Let $\Phi$ be the q-c
channel (\ref{qc}), then the inequality
\begin {equation}
 H_{\Phi\otimes \Omega}(\rho)\ge  H_{\Phi}(\Tr_{{\mathcal K}}(\rho))+ H_{\Omega}(\Tr_{{\mathcal H}}(\rho)),\nonumber
 \label{strongqc}
\end {equation}
 holds for an
arbitrary quantum channel $\Omega :{\mathfrak
S}({\mathcal K})\to {\mathfrak S}({\mathcal K})$.
\label{theoqc}
\end{theorem}

To prove the theorem we need of the following lemma.

\begin{lemma}
Let $\Phi $ be the q-c channel (\ref {qc}).
Then, given a state $ \rho \in \mathfrak{S}({\mathcal H}\otimes {\mathcal K})$, it is
\begin{equation}
S((\Phi \otimes Id)( \rho))\ge S(\Phi (\Tr_{{\mathcal K}}( \rho
)))+\sum \limits _{j=1}^{d}\lambda _{j}S( \rho _{j}),
\nonumber
\end{equation}
where $\lambda _{j}=\Tr( M_{j}\Tr_{{\mathcal K}}( \rho)),\  \rho
_{j}=\frac {1}{\lambda _{j}}\Tr_{{\mathcal H}}(( M_{j}\otimes I_{{\mathcal K}})
\rho )\in \mathfrak{S}({\mathcal K})$.
\label{lemqc}
\end{lemma}

\begin{proof} [Proof Lemma \ref{lemqc}]
Let us define a quantum channel $\Sigma _{ \rho}:\mathfrak{S}
({\mathcal H})\to \mathfrak {S}({\mathcal H}\otimes {\mathcal K})$ by the formula
\begin{equation}
\Sigma _{ \rho}( \sigma):=\sum \limits
_{j=1}^{d}\Tr(|e_{j}\rangle\langle e_{j}| \sigma )|e_{j}\rangle\langle e_{j}|\otimes
 \rho _{j},
 \nonumber
\end{equation}
where the states $ \rho _{j}\in \mathfrak{S}({\mathcal K})$ are the same
as in the formulation of the Lemma \ref{lemqc}. One can see that
\begin {equation}\label {f1}
\Sigma _{ \rho}(\Phi (\Tr_{{\mathcal K}}( \rho)))=(\Phi \otimes
Id)( \rho),
\end {equation}
\begin {equation}\label {f2}
\Sigma _{ \rho}\left(\frac {1}{d}I_{{\mathcal H}}\right)=\frac {1}{d}\sum \limits
_{j=1}^{d}|e_{j}\rangle\langle e_{j}|\otimes  \rho _{j}.
\end {equation}

The decreasing property of the relative entropy Eq.(\ref {decr})
gives us
\begin{equation}
S\left(\Sigma _{ \rho }(\Phi (\Tr_{{\mathcal K}}( \rho)))\,\Big\|\,\Sigma _{
\rho }(\frac {1}{d}I_{{\mathcal H}})\right)\le S\left(\Phi (\Tr_{{\mathcal K}}( \rho))\,\Big\|\,\frac
{1}{d}I_{{\mathcal H}}\right).
\nonumber
\end{equation}
Taking into account Eq.(\ref {f1}) and (\ref {f2}) we get
\begin{equation}
S\left(\Phi (\Tr_{{\mathcal K}}( \rho))\,\Big\|\,\frac {1}{d}I_{{\mathcal H}}\right)=\log d-S\left(\Phi
(\Tr_{{\mathcal K}}( \rho))\right),
\nonumber
\end{equation}
and
\begin{eqnarray}
S\left(\Sigma _{ \rho }(\Phi(\Tr_{\mathcal K} ( \rho)))\,\Big\|\,\Sigma _{ \rho
}(\frac {1}{d}I_{{\mathcal H}})\right)&=&\log d+\sum \limits _{j=1}^{d}\lambda
_{j}S( \rho _{j})\nonumber\\
&-&S((\Phi \otimes Id)( \rho)),
\nonumber
\end{eqnarray}
from which the result of Lemma \ref{lemqc} follows.
\hfill\end{proof}

\bigskip

\begin{proof}[Proof Theorem \ref{theoqc}]
Let $\Phi $ be the q-c channel (\ref {qc}). Suppose that $\Omega$
is an arbitrary channel and
\begin {equation}\label {avrqc}
 \rho =\sum \limits _{j=1}^{k}p_{j} \rho _{j},
\end {equation}
such that the states $ \rho _{j},\ 1\le j\le k,$ form the
optimal ensemble for the output entropy of $\Phi\otimes\Omega$, i.e.
\begin {equation}\label {avr2qc}
 H_{\Phi \otimes \Omega}( \rho)=\sum \limits
_{j}p_{j}S((\Phi \otimes \Omega)( \rho _{j})).
\nonumber
\end {equation}
Applying Lemma \ref{lemqc} to each term in the sum on the right hand
side we get
\begin{eqnarray}
 H_{\Phi \otimes \Omega}( \rho )&\ge& \sum \limits
_{j}p_{j}S(\Phi (\Tr_{{\mathcal K}}( \rho_{j})))\nonumber\\
&+&\sum \limits _{j}p_{j}\sum \limits _{k=1}^{d}\lambda
_{jk}S(\Omega ( \rho _{jk})),
\nonumber
\end{eqnarray}
where $\lambda _{jk}=\Tr( M_{k}\Tr_{{\mathcal K}}( \rho _{j}))$ and
$\rho _{jk}=\frac {1}{\lambda _{jk}}\Tr_{{\mathcal H}}(( M_{k}\otimes
I_{{\mathcal K}}) \rho _{j})\in \mathfrak{S}({\mathcal K})$. By the definitions \eqref{avrqc}
and \eqref{quant} we obtain on the one hand
\begin{equation}
\sum \limits _{j}p_{j}S(\Phi (\Tr_{{\mathcal K}}( \rho_{j})))\ge
H_{\Phi}(\Tr_{{\mathcal K}}( \rho )).
\nonumber
\end{equation}
On the other hand,
\begin{equation}
\sum \limits _{j}p_{j}\sum \limits _{k}\lambda _{jk}\Omega ( \rho
_{jk})=\Omega (\Tr_{{\mathcal H}}( \rho )).
\nonumber
\end{equation}
The last formula implies that
\begin{equation}
\sum \limits _{j}p_{j}\sum \limits _{k=1}^{d}\lambda _{jk}S(\Omega
( \rho _{jk}))\ge  H_{\Omega}(\Tr_{{\mathcal H}}( \rho )).
\nonumber
\end{equation}
Then the result of Theorem \ref{theoqc} follows.
\hfill\end{proof}

\bigskip

Notice that a q-c channel is a partial case of the
entanglement-breaking channels considered in \cite{Hol04}.
 So our proof is alternative to the one given in \cite{Hol04}
 for entanglement-breaking channels.


\subsection {The quantum erasure channel}

Let ${\mathcal H}$ and ${\mathcal H}'$ be Hilbert spaces of dimension $d$ and $d+1$
respectively. We claim that ${\mathcal H}\subset {\mathcal H}'$ which results in the
inclusion $\mathfrak{S}({\mathcal H})\subset \mathfrak{S}({\mathcal H}')$. Suppose that
$|\omega\rangle \in {\mathcal K}$ is orthogonal to ${\mathcal H}$. Fix $\epsilon$ such that   $0\le \epsilon
\le1$, then we call quantum erasure channel the CP-map $\Phi:\mathfrak{S}({\mathcal H})\to
\mathfrak{S}({\mathcal H}')$ defined by
\begin{equation}\label {er}
\Phi ( \rho ):=\epsilon |\omega\rangle\langle \omega |+(1-\epsilon ) \rho ,
\end{equation}
with $ \rho \in \mathfrak{S}({\mathcal H})$.
Notice that this is a generalization to dimension $d$ of the qubit erasure channel introduced in \cite{Ben97}.

\begin{theorem}
Let $\Phi$ be the erasure
channel (\ref{er}), then the
inequality
\begin {equation}
 H_{\Phi\otimes \Omega}(\rho)\ge  H_{\Phi}(\Tr_{{\mathcal K}}(\rho))+ H_{\Omega}(\Tr_{{\mathcal H}}(\rho)),
 \label{stronger}
 \nonumber
\end {equation}
 holds for an
arbitrary quantum channel $\Omega :{\mathfrak
S}({\mathcal K})\to {\mathfrak S}({\mathcal K})$.
\label{theoer}
\end{theorem}

To prove the theorem we need of the following lemma.

\begin{lemma}
Let $\Phi $ be the quantum erasure channel \eqref{er}.
Then, given a state $ \rho \in \mathfrak{S}({\mathcal H}\otimes {\mathcal K})$ it is
$$
S((\Phi \otimes Id)( \rho))\ge \epsilon S(\Tr_{{\mathcal H}}( \rho
))+(1-\epsilon )S( \rho )+S(\Phi (\Tr_{{\mathcal K}}( \rho ))).
$$
\label{lemer}
\end{lemma}

\begin{proof} [Proof Lemma \ref{lemer}]
Denote by $P_{{\mathcal H}}$ the orthogonal projection in ${\mathcal H}'$ onto the
subspace ${\mathcal H}$. Given $ \rho \in \mathfrak{S}({\mathcal H}\otimes {\mathcal K})$ let us
define a quantum channel $\Sigma _{ \rho}:\mathfrak{S} ({\mathcal H}')\to
\mathfrak{S}({\mathcal H}'\otimes {\mathcal K})$ by the formula
\begin{equation}
\Sigma _{ \rho}( \sigma):=\Tr(|\omega\rangle\langle\omega | \sigma
)|\omega \rangle\langle\omega |\otimes \Tr_{{\mathcal H}}( \rho)+\Tr(P_{{\mathcal H}}
\sigma ) \rho,
\nonumber
\end{equation}
with $ \sigma \in \mathfrak{S}({\mathcal H}')$.

Pick up the orthogonal projection $|e \rangle\langle e|$ from the spectral
decomposition of the state $\Tr_{{\mathcal K}}( \rho )$. One can see that
\begin {equation}\label {f3}
\Sigma _{ \rho}(\Phi (\Tr_{{\mathcal K}}( \rho)))=(\Phi \otimes
Id)( \rho),
\end {equation}
\begin {equation}\label {f4}
\Sigma _{ \rho}\left(\frac {1}{2}|\omega \rangle\langle \omega |+\frac
{1}{2}|e\rangle\langle e|\right)=\frac {1}{2}|\omega \rangle\langle
\omega |\otimes \Tr_{{\mathcal H}}( \rho )+\frac {1}{2} \rho.
\end {equation}
 The decreasing
property of the relative entropy (\ref {decr}) gives us
\begin{eqnarray}
&&S\left(\Sigma _{ \rho }\left(\Phi (\Tr_{{\mathcal K}}(
\rho))\right)\,\Big\|\, \Sigma _{ \rho }\left(\frac {1}{2}|\omega
\rangle\langle \omega |+\frac {1}{2}|e\rangle\langle
e|\right)\right)
\nonumber\\
&&\le S\left(\Phi (\Tr_{{\mathcal K}}( \rho))\,\Big\|\, \frac {1}{2}|\Omega \rangle\langle \Omega |+\frac
{1}{2}|e\rangle\langle e|)\right).
\nonumber
\end{eqnarray}
Taking into account (\ref {f3}) and (\ref {f4}) we get
\begin{eqnarray}
&&S\left(\Phi (\Tr_{{\mathcal K}}( \rho))\,\Big\|\,\frac {1}{2}|\omega
\rangle\langle\omega |+\frac
{1}{2}|e \rangle\langle e|\right)\nonumber\\
&&=(\epsilon +(1-\epsilon)\langle e|\Tr_{{\mathcal K}}( \rho)|e\rangle)\log d-S(\Phi
(\Tr_{{\mathcal K}}( \rho)))\nonumber\\
&&\le \log d-S(\Phi (\Tr_{{\mathcal K}}( \rho))),
\nonumber
\end{eqnarray}
and
\begin{eqnarray}
&&S\left(\Sigma _{ \rho }(\Phi ( \rho))\,\Big\|\,\Sigma _{
\rho}(\frac {1}{2}|\omega \rangle\langle\omega |+\frac {1}{2}|e
\rangle\langle e|)\right)=
\nonumber\\
&&\log d+\epsilon S(\Tr_{{\mathcal H}}( \rho ))+(1-\epsilon )S(
\rho)-S((\Phi \otimes Id)( \rho)).\nonumber
\end{eqnarray}
The result of Lemma \ref{lemer} then follows.
\hfill\end{proof}

\bigskip

\begin{proof} [Proof Theorem \ref{theoer}]
Let $\Phi $ be the erasure channel (\ref {er}). Suppose that $\Omega$ is an arbitrary channel and
\begin {equation}\label {avrer}
 \rho =\sum \limits _{j=1}^{k}p_{j} \rho _{j}
\end {equation}
is such that the states $ \rho _{j},\ 1\le j\le k,$ form the
optimal ensemble for for the output entropy of $\Phi\otimes\Omega$, i.e.
\begin {equation}\label {avr2er}
 H_{\Phi \otimes \Omega}( \rho)=\sum \limits
_{j}p_{j}S((\Phi \otimes \Omega)( \rho _{j}))=\sum \limits
_{j}p_{j}S((\Phi \otimes Id)(\tilde \rho _{j})),
\nonumber
\end {equation}
with $\tilde \rho _{j}=(Id\otimes \Omega)(\rho _{j})$.
Applying Lemma \ref{lemer} to each term in the sum on the right hand
side of the above equation we get
\begin{eqnarray}
 H_{\Phi \otimes \Omega}( \rho )\ge \sum \limits
_{j}&p_{j}&  \left [\epsilon S(\Omega (\Tr_{{\mathcal H}}( \rho
_{j})))\right.\nonumber\\
&&+\left.(1-\epsilon )S((Id\otimes \Omega)( \rho _{j}))\right.
\nonumber\\
&&+\left. S(\Phi
(\Tr_{{\mathcal K}}( \rho _{j} )))\right ].
\nonumber
\end{eqnarray}
Notice also that
$$
\sum \limits _{j}p_{j}S((Id\otimes \Omega)( \rho _{j}))\ge
H_{Id\otimes \Omega}( \rho)\ge  H_{\Omega}(\Tr_{{\mathcal H}}( \rho ))
$$
because the strong superadditivity conjecture holds for the
noiseless channel \cite{Hol04}. Then, the result of Theorem \ref{theoer} follows.
\hfill\end{proof}


\subsection {The quantum depolarizing channel}

The quantum depolarizing channel in the
Hilbert space ${\mathcal H}$ of dimension $d$
is defined as \cite{Amo07a}
\begin{equation}
\Phi (\rho):=(1-p)\rho+\frac {p}{d}I_{{\mathcal H}}, \label{qdc}
\end{equation}
with $\ \rho\in \mathfrak{S} ({\mathcal H}),\ 0\le p\le d^{2}/(d^{2}-1)$.

\begin{theorem}
Let $\Phi$ be the quantum depolarizing
channel (\ref{qdc}), then the inequality
\begin {equation}
 H_{\Phi\otimes \Omega }(\rho)\ge  H_{\Phi}(\Tr_{{\mathcal K}}(\rho))+ H_{\Omega }(\Tr_{{\mathcal H}}(\rho)),
 \nonumber
\end {equation}
 holds for an
arbitrary quantum channel $\Omega :{\mathfrak
S}({\mathcal K})\to {\mathfrak S}({\mathcal K})$.
\label{dc}
\end{theorem}

To prove Theorem \ref{dc} we need of some properties of the quantum depolarizing channel.

 Following Ref.\cite{Kin02}, by choosing an orthonormal basis $\{|f_{j}\rangle\}$ in ${\mathcal H}$,
we can define a set of orthonormal bases
$\{\{|e^k_j\rangle\}_{j=0}^{d-1}\}_{k=1}^{2d^{2}}$ as
\begin{equation}\label {bas}
|e_{j}^{k}\rangle:=\sum \limits _{s=0}^{d-1}\exp\left(i\frac {2\pi
s^{2}k}{2d^{2}}\right)\exp\left(i\frac {2\pi j}{d}\right)|f_{s}\rangle,
\end{equation}
with $1\le k\le 2d^{2}$. Moreover, let
\begin{eqnarray}
U&:=&\sum \limits _{s=0}^{d-1}\exp\left(i\frac {2\pi s}{d}\right)|f_{s}\rangle\langle f_{s}|,\nonumber\\
V_{k}&:=&\sum \limits _{s=0}^{d-1}\exp\left(i\frac {2\pi
s}{d}\right)|e_{s}^{k}\rangle\langle e_{s}^{k}|,\nonumber
\end{eqnarray}
be unitary operators in ${\mathcal H}$. We introduce
phase damping channels as follows
\begin{equation}\label{kanaly}
\Psi _{k}(\rho )=\left(1-\frac {d-1}{d}p\right)\rho +\frac {p}{d}\sum
\limits _{s=1}^{d-1}V_{k}^{s}\rho V_{k}^{s},
\nonumber
\end{equation}
with $\rho \in \mathfrak{S}({\mathcal H}),\ 1\le k\le 2d^{2}$.

Then, the quantum depolarizing $\Phi$ can be expressed in terms of the above phase damping channels as
\begin{eqnarray}
\Phi (\rho)&=&\frac {1-p}{1+(d-1)(1-p)}\frac {1}{2d}\sum
\limits _{k=1}^{2d^{2}}\Psi _{k}(\rho)\nonumber\\
&+&\frac {p}{1+(d-1)(1-p)}\frac {1}{2d^{3}}\sum \limits
_{j=1}^{d-1}\sum \limits _{k=1}^{2d^{2}}U^{j}\Psi _{k}(\rho
)U^{*j},\nonumber\\
\label {king}
\end{eqnarray}
with $\rho \in \mathfrak{S}({\mathcal H})$.
By defining
\begin{equation}\label {expecta}
E_{k}(\rho):=\frac {1}{d}\sum \limits _{s=0}^{d-1}V^{s}_{k}\rho
U^{*s}_{k},
\nonumber
\end{equation}
the conditional expectations on the algebras of fixed elements for the phase
dampings $\Psi _{k}$, we have
\begin {equation}\label {spur}
E_{k}(|f_{j}\rangle\langle f_{j}|)=\frac {1}{d}I_{{\mathcal H}},
\nonumber
\end {equation}
for $1\le k\le 2d^{2},\ 0\le j\le d-1$. This property guarantees
that the basis $\{|f_{j}\rangle\}$ is mutually unbiased with respect to all
the bases $\{|e^k\rangle\}$ defined by (\ref {bas}).
\bigskip


\begin{proof}[Proof Theorem \ref{dc}]
Let us take the optimal ensemble corresponding to the state $\rho $ such
that
$$
H_{\Phi \otimes \Omega}(\rho)=\sum \limits _{s}\pi _{s}S((\Phi
\otimes \Omega )(\rho _{s})).
$$
In the following we shall estimate $S((\Phi \otimes \Omega )(\rho
_{s}))$ for each fixed $s$.

Let us consider  for a while $\varrho$ instead of a $\rho_s$.
Let us pick up a unitary operator $T$ such that the state
\begin{equation}
\tilde \varrho =(T\otimes I_{{\mathcal K}})(Id\otimes \Omega )(\varrho
)(T^{*}\otimes I_{{\mathcal K}})), \nonumber
\end{equation}
satisfies the property
\begin {equation}\label {for}
E_{k}(\Tr_{{\mathcal K}}(\tilde \varrho ))=\frac {1}{d}I_{{\mathcal H}}.
\nonumber
\end {equation}

Using the covariance property $\Phi (\sigma )=T^{*}\Phi (T\sigma
T^{*})T$, taking place for all states $\sigma \in \mathfrak{S}({\mathcal H})$, we can rewrite the
decomposition (\ref {king}) as follows
\begin{eqnarray}
&&\Phi (\sigma)=\frac {1-p}{1+(d-1)(1-p)}\frac {1}{2d}\sum
\limits _{k=1}^{2d^{2}}\tilde \Psi _{k}(\sigma)\nonumber\\
&&+\frac {p}{1+(d-1)(1-p)}\frac {1}{2d^{3}}\sum \limits
_{j=1}^{d-1}\sum \limits _{k=1}^{2d^{2}}T^{*}U^{j}T\tilde \Psi
_{k}(\sigma )TU^{*j}T^{*},\nonumber\\
\label {king+}
\end{eqnarray}
where $\tilde \Psi _{k}(\sigma )=T^{*}\Psi _{k}(T\sigma T^{*})T$
are the phase damping channels with the property
\begin {equation}\label{ll}
\Tr(\tilde E_{k}(\Tr_{{\mathcal K}}(\varrho )))=\frac {1}{d}I_{{\mathcal H}}.
\end {equation}
Here $\tilde E_{k}(\sigma )=T^{*}E_{k}(T\sigma T^{*})T$, $\varrho
\in \mathfrak {S}({\mathcal H}\otimes {\mathcal K})$  and $\sigma
\in \mathfrak {S}({\mathcal H})$. The above equality guarantees that the state
$\varrho$ is unbiased with respect to all the orthonormal bases which
form the unitary operators determining the action of the phase
damping channels $\Psi _{k},\ 0\le k\le d-1$.

It follows from the decomposition (\ref {king+}) that
\begin {equation}\label {l}
S((\Phi \otimes \Omega )(\varrho ))\ge \frac {1}{d}\sum \limits
_{k=0}^{d-1}S((\tilde \Psi _{k}\otimes \Omega )(\varrho )). \nonumber
\end {equation}
Applying Theorem \ref {fff2} to each term of the sum in the right
hand side and taking into account that $-\sum \limits
_{j=0}^{d-1}\lambda _{j}\log\lambda _{j}=H_{\Phi}(\Tr_{{\mathcal
K}}(\varrho ))$ for $\lambda _{0}=1-\frac {d-1}{d}p,\ \lambda
_{j}=\frac {p}{d},\ 1\le j\le d-1$ due to (\ref{ll}), we get
\begin {equation}
S((\Phi \otimes \Omega )(\rho_s ))\ge
H_{\Phi}(\Tr_{{\mathcal K}}(\rho_s))+ H_{\Omega }(\Tr_{{\mathcal H}}(\rho_s)),
\nonumber
\end{equation}
hence the result of Theorem \ref{dc}. \hfill\end{proof}


\section {Conclusion}

By using the decreasing property of the relative entropy, we have
proved the strong superadditivity for a class of quantum channels.
This class includes the channels for which the property was
already shown by using other methods (thus giving an alternative
proof) as well as others channels (thus providing an extension of
the class). We guess that the decreasing property of the relative
entropy could be a powerful tool for a further extension of such
class of channels. More generally, it  could constitute a
universal method to investigate relevant properties of memoryless
quantum channels. In fact,  as a fall down of the strong
superadditivity property we get the additivity property. Thus for
our class of channels, the additivity results automatically
proved.

The perspective of a global proof of additivity through strong superadditivity seems fascinating
and motivate further investigations, especially in consideration of the limits of other methods \cite{Hay07}.

\section*{Acknowledgment}
We thank David Gross for useful remarks. The work of G.G. Amosov
is partially supported by INTAS grant Nr. 06-1000014-6077. The
work of S. Mancini is partially supported by the European
Commision under the Integrated Projects QAP and SCALA.



\begin{thebibliography}{99}


\bibitem{Ved02}
V. Vedral, ``The role of relative entropy in quantum information
theory,'' \emph{Rev. Mod. Phys.} vol. 74. pp. 197-234, 2002.

\bibitem{Lin75}
G. Lindblad, ``Completely positive maps and entropy
inequalities,'' \emph{Comm. Math. Phys.} vol. 40, pp. 147-151,
1975.

\bibitem{Hol72}
Holevo A. S. , ``On the mathematical theory of quantum
communication channels,'' \emph{Probl. Inf. Transm.} vol. 8, pp.
62-71, 1972.

\bibitem {Hol98}
A. S. Holevo, ``Quantum coding theorems,'' \emph{Russ. Math.
Surveys}, vol. 53, pp. 1295-1331, 1998.

\bibitem{Sch97}
B. Schumacher, and M.D. Westmoreland, ``Sending classical
information via noisy quantum channels,'' \emph{Phys. Rev. A}
vol. 56, pp. 131-138, 1997.

\bibitem {Amo00}
G.G. Amosov, A.S. Holevo, and R.F. Werner, ``On some additivity
problems in quantum information theory,'' \emph{Probl. Inf.
Transm.} vol. 36, pp. 305-313, 2000.

\bibitem{Kin02}
C. King,
``The capacity of the quantum depolarizing channel,''
\emph{IEEE Trans. Inf. Th.} vol. 49, pp. 221-???, 2003.

\bibitem{Amo07b}
G.G. Amosov, ``On the Weyl  channels being covariant with respect
to the maximum commutative group of unitaries,'' \emph{J. Math.
Phys.} vol. 48, pp. 012104-01--012104-14, 2007.

\bibitem{Amo06a}
G.G. Amosov, ``Remark on the additivity conjecture for the
depolarizing quantum channel,'' \emph{Probl. Inf. Transm.} vol.
42, pp. 69-76, 2006.

\bibitem{Hay07}
P. Hayden, ``The maximal p-norm multiplicativity conjecture is false'',
Available: http://arxiv.org/abs/0707.3291

\bibitem{Hol04}
A.S. Holevo, and M.E. Shirokov, ``On Shor's channel extension and
constrained channels,'' \emph{Commun. Math. Phys.} vol. 249, pp.
417-430, 2004.

\bibitem{Amo07a}
G.G. Amosov, ``The strong superadditivity conjecture holds for a
quantum depolarizing channel in any dimension,'' \emph{Phys. Rev.
A} vol. 75, pp. 060304-1--060304-2, 2007.

\bibitem{Ohy93}
M. Ohya, and D. Petz, \emph{ Quantum Entropy and Its Use,
Texts and Monographs in Physics}, Berlin, Springer-Verlag, 1993.

\bibitem{Ume62}
H. Umegaki, ``Conditional expectation in an operator algebra. IV.
Entropy and information,'' \emph{Kodai Math. Sem. Rep.} vol. 14,
pp. 59-85, 1962.

\bibitem{Hia91}
F. Hiai, and D. Petz, ``The proper formula for relative entropy
and its asymptotics in quantum probability,'' \emph{Comm. Math.
Phys.} vol. 143, pp. 99-114, 1991.

\bibitem{Cov91}
T.M. Cover, and J.A. Thomas,  \emph{Elements of Information
Theory}, New York, Wiley-Interscience Publication, 1991.

\bibitem{Iva81}
I.D. Ivanovich, ``Geometrical description of quantum state
determination,'' \emph{J. Phys. A} vol. 14, pp. 3241-3245, 1981.

\bibitem{Ben97}
C.H. Bennett, D.P. DiVincenzo, and  J.A. Smolin, ``Capacities of
quantum erasure channels,'' \emph{Phys. Rev, Lett.} vol. 78, pp.
3217-3220, 1997.




\end{thebibliography}
\end{document}